\documentclass[conference]{IEEEtran}
 \usepackage[letterpaper, left=0.67in, right=0.67in, bottom=1.02in, top=0.725in]{geometry}
\IEEEoverridecommandlockouts
\usepackage{amsmath,amsfonts}
\usepackage{amsthm}
\usepackage{array}
\usepackage{textcomp}
\usepackage{stfloats}
\usepackage{url}
\usepackage{microtype}
\usepackage{verbatim}
\usepackage{graphicx}
\usepackage{cite}
\usepackage{enumitem}
\usepackage{hyperref}
\usepackage{subcaption} %
\usepackage[linesnumbered,lined,algoruled,commentsnumbered]{algorithm2e}

\usepackage{caption}
\usepackage{subcaption}
\usepackage[font=small,skip=0pt]{caption}
\usepackage[export]{adjustbox}
\usepackage[inkscapelatex=false]{svg}

\SetCommentSty{mycommfont}
\theoremstyle{plain}
\newtheorem{theorem}{Theorem}

\usepackage{units}
\newcommand{\nth}[1]{{#1}^{\text{th}}}

\newcommand{\abs}[1]{\left|{#1}\right|}

\newcommand{\NBS}[0]{N_{\mathrm{BS}}}

\newcommand{\RD}[0]{r_{\mathrm{RD}}}
\newcommand{\au}[0]{\theta_{\mathrm{u}}}

\newcommand{\rf}[0]{r_{\mathrm{F}}}

\begin{document}

\title{Near-Field Beam Prediction Using Far-Field Codebooks in Ultra-Massive MIMO Systems \\
\author{Ahmed Hussain, Asmaa Abdallah, Abdulkadir Celik, and Ahmed M. Eltawil,
\\ Computer, Electrical, and Mathematical Sciences and Engineering (CEMSE) Division,
\\King Abdullah University of Science and Technology (KAUST), Thuwal, 23955-6900, KSA }
}

\maketitle

\begin{abstract}
Ultra-massive multiple-input multiple-output (UM-MIMO) technology is a key enabler for 6G networks, offering exceptional high data rates in millimeter-wave (mmWave) and Terahertz (THz) frequency bands. The deployment of large antenna arrays at high frequencies transitions wireless communication into the radiative near-field, where precise beam alignment becomes essential for accurate channel estimation. Unlike far-field systems, which rely on angular domain only, near-field necessitates beam search across both angle and distance dimensions, leading to substantially higher training overhead. To address this challenge, we propose a discrete Fourier transform (DFT) based beam alignment to mitigate the training overhead. We highlight that the reduced path loss at shorter distances can compensate for the beamforming losses typically associated with using far-field codebooks in near-field scenarios. Additionally, far-field beamforming in the near-field exhibits angular spread, with its width determined by the user's range and angle. Leveraging this relationship, we develop a correlation interferometry (CI) algorithm, termed CI-DFT, to efficiently estimate user angle and range parameters. Simulation results demonstrate that the proposed scheme achieves performance close to exhaustive search in terms of achievable rate while significantly reducing the training overhead by 87.5\%.
\end{abstract}

\begin{IEEEkeywords}
Near-field, beam training, angular spread, correlative interferometry
\end{IEEEkeywords}

\section{Introduction}
Future networks operating in the millimeter-wave (mmWave) and terahertz (THz) frequency bands are anticipated to provide substantial bandwidth, enabling extremely high-capacity wireless links \cite{rappaport2019wireless}. However, these higher frequencies face considerable challenges due to significant path loss and molecular absorption effects \cite{han2021hybrid}. Conversely, the shorter wavelengths associated with these bands make it feasible to deploy ultra-massive multiple-input-multiple-output (UM-MIMO) antenna arrays within a compact footprint. The combination of smaller wavelengths and larger array apertures will likely result in future communications predominantly operating in the radiative near-field.

Highly directional links are achieved due to the combination of high operating frequency and large antenna aperture. The resulting narrow beams impose strict requirements on precise beam alignment. Beam training and beam estimation techniques are commonly employed to determine the optimal beam direction \cite{giordani2018tutorial,OS_Asmaa}. Beam training specifically involves sweeping a grid of beams to cover all potential users. After beam estimation, the user reports the beam index corresponding to the direction with the highest gain.

The challenge of beam alignment in near-field communications stems from the high training overhead, making the process highly resource-intensive \cite{asmaa2024near}. Unlike far-field, near-field beam training involves sweeping beams across both angular and distance domains. Considering the hardware cost and complexity of deploying a radio frequency (RF) chain for each antenna, UM-MIMO systems typically employ hybrid beamforming (HBF) architectures. The cost and complexity reduction comes at the expense of increased beam training overhead since the number of beams that can be generated simultaneously by HBF is limited by the available RF chains.

Recently, various codebook design strategies have been explored to address the challenges of high complexity and training overhead of UM-MIMO systems. A polar codebook with non-uniform distance sampling is presented in \cite{cui2022channel}. One straightforward approach with this codebook involves performing an exhaustive search across the angle-distance domain. To mitigate the search complexity, a two-step method is used: first, the user’s angle is estimated with a DFT codebook, followed by range estimation using the polar-domain codebook \cite{zhang2022fast}. Another strategy entails developing a multi-resolution codebook \cite{lu2023hierarchical}, drawing inspiration from hierarchical beam training techniques used in far-field scenarios, although these methods require user feedback to select the next set of codewords. Subarray-based schemes \cite{wu2023two} initially establish a coarse user direction by assuming far-field conditions and then refine it using a specialized near-field codebook in the polar domain. However, maintaining the far-field assumption requires reducing subarray element count, leading to broader beams and reduced spatial resolution. Moreover, existing methods that utilize polar codebook-based beam training suffer from excessive training overhead. 

To reduce the training overhead, we propose a near-field beam training method using a DFT codebook, challenging the prevailing reliance on polar codebook approaches in the literature. The DFT codebook in the near-field leads to energy dispersion across the angular domain, producing multiple peaks within the beam space—a phenomenon referred to as \textit{angular spread}. Traditionally, the angular energy spread has been viewed as a limitation, favoring the use of polar codebooks. However, we capitalize on this angular spread to extract both range and angle information for effective near-field beam training. Furthermore, we introduce the coverage regions for both near-field and far-field beams and show that far-field beams can maintain adequate signal-to-noise ratio (SNR) in the near-field by balancing the reduced beamforming gain with the decreased path loss at closer distances. We generate a grid of orthogonal beams based on the DFT codebook, revealing that a near-field user experiences significant gains from a specific set of DFT beams, constituting the angular spread. We show that the width of this angular spread is related to the user's location, and leverage this relationship to develop a \textit{correlation interferometry} (CI) algorithm called CI-DFT, which effectively estimates user angle and range parameters. Simulation results confirm the superior performance of our proposed CI-DFT algorithm compared to polar codebook-based methods, achieving an 87.5\% reduction in beam training overhead compared to exhaustive search.



\begin{figure}[t]
\centering
\includegraphics[width=0.35\textwidth]{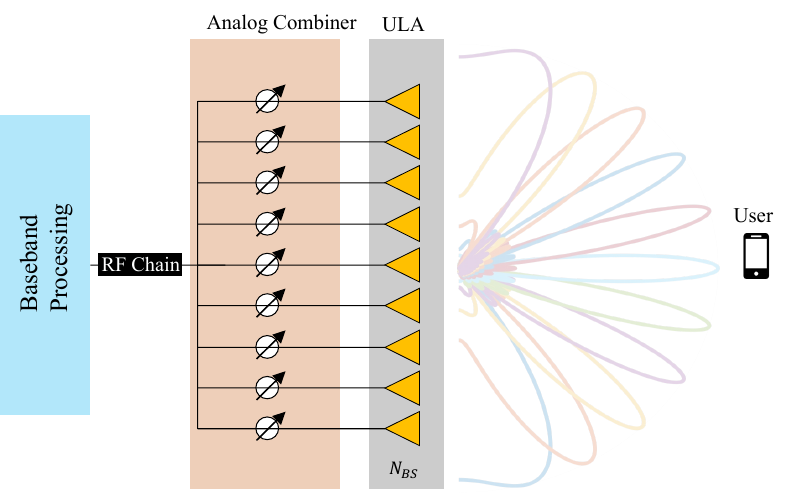}
\setlength{\belowcaptionskip}{-15pt}
\caption{
UM-MIMO System Model }
\label{system_model}
\end{figure}

\section{System Model and Problem Formulation } \label{System Model and Problem Formulation}
We consider the downlink beam training for a narrow-band communication system with UM-MIMO antenna array at the base station (BS) and a single isotropic antenna user equipment (UE) as depicted in Fig. \ref{system_model}. We consider a single RF chain transceiver with a uniform linear array (ULA) of $\NBS$ antenna elements spaced $d = \lambda/2$ apart to avoid grating lobes. 

\subsection{Channel Models}
 The far-field channel model is characterized by the spatial steering vector, $\mathbf{a}(\theta) $, that depends on angle, $\theta$, only. The far-field multipath channel $\mathbf{h}^{\text{far-field}} \in \mathbb{C}^{\NBS \times 1}$, is modeled as 
\begin{equation}
\mathbf{h}^{\text{far-field}} = \sqrt{\frac{\NBS}{L}}\sum_{l=1}^{L}g_le^{-j{\nu}r_l}\mathbf{a}(\theta_l),
\label{eqn1}
\end{equation}
where $L$ denotes the number of channel paths; $g_l$ represents the complex gain of the $\nth{l}$ path; $\theta_l$ is the angle of arrival (AoA); $r_l$ is the distance between the user and ULA along the $\nth{l}$ channel path; \( \nu = \frac{2\pi f}{c} \) is the wavenumber; \( f \) is the carrier frequency; and \( c \) is the speed of light. Based on the planar wavefront assumption, the spatial steering vectors $\mathbf{a}(\theta) \in \mathbb{C}^{\NBS \times 1}$ are formulated as
\begin{equation} 
\mathbf{a} (\theta) = \tfrac{1}{\sqrt{\NBS}}\Big[1, e^{-j\nu d\sin(\theta)} ,\dots, e^{-j\nu d(\NBS-1)\sin{(\theta)}}\Big]^\mathsf{T}.
\label{eqn2}
\end{equation}

On the other hand, the near-field channel, $\mathbf{h} \in \mathbb{C}^{\NBS \times 1}$, relies on spherical wavefront assumption, which results in the following channel model 
\begin{equation} 
\mathbf{h} = \sqrt{\frac{\NBS}{L}}\sum_{l=1}^{L}g_le^{-j{\nu}r_l}\mathbf{b}(\theta_l,r_l),
\label{eqn3}
\end{equation}
where $\mathbf{b}(\theta_l,r_l) \in \mathbb{C}^{\NBS \times 1}$ is the near-field steering vector that depends on both distance $r_l$ and angle $\theta_l$. The near-field steering vector is given by
\begin{equation} 
\mathbf{b} (\theta_l,r_l) = \tfrac{1}{\sqrt{\NBS}}\Big[e^{-j\nu_{c}(r_l^{(0)} -r_l)} ,\dots, e^{-j\nu_{c}(r_l^{(\NBS-1)} -r_l)}\Big]^\mathsf{T},
\label{eqn4}
\end{equation}
where $r^{(n)}$ is the distance between the user and the $n^{th}$ antenna element; the phase shift factor $\nu(r^{(n)}-r)$ is given by $\frac{2\pi}{\lambda} (\sqrt{r^2+ n^2d^2-2rnd\sin{(\theta)}} - r)$ \cite{cui2024near}, where non-linear phase progression can be observed.

\subsection{Problem Formulation}
During downlink training, $\nth{p}$ pilot transmitted by the BS is denoted as $x_p$, $p \in [1,P]$, where $P$ is the total pilot length. Then, the received signal by the user at the $\nth{p}$ pilot is expressed as 
\begin{equation} 
{y}_p= \mathbf{w}_{p
}\mathbf{h}x_p + \mathbf{w}_{p}\mathbf{n},
\label{eqn5}
\end{equation}
where $\mathbf{w}_{p} \in\mathbb{C} ^ {1 \times \NBS}$ is the analog combining vector, $\mathbf{n} \sim \mathcal{N}(0,\sigma^2)$ is the complex Gaussian noise, and $\sigma^2$ is the noise variance. Since the high-frequency narrowband communication between UM-MIMO transceivers is dominated by the single line-of-sight (LoS) path, this work focuses on achieving beam alignment based on the LoS path. Therefore, we focus on the near-field LoS channel, $\mathbf{h}(\theta,r)$, in the rest of the paper. Beam alignment aims to find the optimal codeword that maximizes the SNR
\begin{equation} 
 \max_{\mathbf{w}} \ \ \left|\mathbf{w}\mathbf{h}(\theta,r) \right|, \\
 \text { s.t. } \left\|\mathbf{w} \right\|_2 = 1, \ \mathbf{w} \in \mathbf{\Phi}, 
\label{eqn6}
\end{equation}
where $\mathbf{\Phi}$ defines a predefined codebook. The optimal solution to \eqref{eqn6} is to select $\mathbf{w}$ from a polar codebook $\mathbf{\Phi}$ 
\begin{equation} 
\boldsymbol{\Phi} =\Big[\boldsymbol{\phi}_1,\cdots,\boldsymbol{\phi}_n,\cdots,\boldsymbol{\phi}_{\NBS}\Big],
\label{eqn7}
\end{equation}
where $\boldsymbol{\phi}_n = \left[\mathbf{b}\left({\theta}_{n}, {r}_{1,n}\right), \mathbf{b}\left({\theta}_{n}, {r}_{2,n}\right), \cdots, \mathbf{b}\left({\theta}_{n}, {r}_{{S_n,n}}\right)\right]$ is constructed by concatenating $S_n$ near-field steering vectors $\mathbf{b}(\theta,r)$ for $S_n$ distance samples corresponding to angle $\theta_n$. Consequently, it is expensive to search over $\NBS S$ codewords in $\mathbf{\Phi} \in \mathbb{C}^{\NBS \times S}$ where $S=\sum_{n=1}^{\NBS} S_n$. Thus, in this paper, we focus on utilizing only $\NBS$ DFT codewords instead of the full set of $\NBS S$ polar codewords. 

\section{Coverage of Near and Far-field Codewords} 
Existing polar codebooks are constructed based on the classical Rayleigh distance, which defines the boundary between the radiative near and far-field and given by $r_\mathrm{RD} = \frac{2D^2}{\lambda}$, where $D$ represents the array aperture. However, the Rayleigh distance tends to overestimate the near-field region concerning beamfocusing and multiplexing gains \cite{bjornson2021primer,hussain2024near,ahmed2024nearfield}. As a result, polar codebooks depending on this boundary incur excessive sampling and leads to large codebook sizes. In the following sections, we revise the more accurate near-field boundaries based on near-field beam focusing from our previous work \cite{hussain2024near,ahmed2024nearfield,asmaa2024near} and discuss the feasibility of using far-field codewords in the near-field region.
 
\subsection{Beamdepth and Effective Beamfocused Rayleigh Distance }
Beamdepth $r_\mathrm{BD}$ is defined as the distance interval, $r \in [\rf^\mathrm{min}, \rf^\mathrm{max}]$, where normalized array gain is at most $\unit[3]{dB}$ lower than its maximum value. For a ULA, the $\unit[3]{dB}$ beamdepth, $r_\mathrm{BD} = \rf^\mathrm{max}-\rf^\mathrm{min}$, obtained by focusing a beam at a distance $\rf$ and angle $\theta$ from the BS is given by \cite{hussain2024near,ahmed2024nearfield}
\begin{equation} 
\begin{aligned}
r_\mathrm{BD} \approx\frac{\rf\RD\cos^2{(\theta)} }{\RD\cos^2{(\theta)}- 10\rf} - \frac{\rf\RD\cos^2{(\theta)} }{\RD\cos^2{(\theta)}+ 10\rf}.
\end{aligned}
\label{eqn8}
\end{equation}

Finite depth beams are achieved in the near-field only when the focus point lies within a specific region in the near-field. The maximum value of $r_\mathrm{BD}$ is obtained when the denominator $\RD^2\cos^4{(\theta)}-100\rf^2 = 0$ in \eqref{eqn8}. Thus, the farthest angle-dependent axial distance $\rf$, where we can have finite depth beamforming, is less than $\frac{\RD}{10}\cos^2{(\theta)}$, termed as effective beamfocused Rayleigh distance (EBRD) \cite{hussain2024near,ahmed2024nearfield,asmaa2024near}. For values of $\rf$ greater than $r_\mathrm{EBRD}$, $r_\mathrm{BD} \rightarrow \infty$. To summarize, near-field beams are achievable within the EBRD region, while far-field beams are formed outside the EBRD. Consequently, the EBRD serves as a practical boundary between the near-field and far-field regions.

\begin{figure}[t]
\centering
\includegraphics[scale = 0.40]{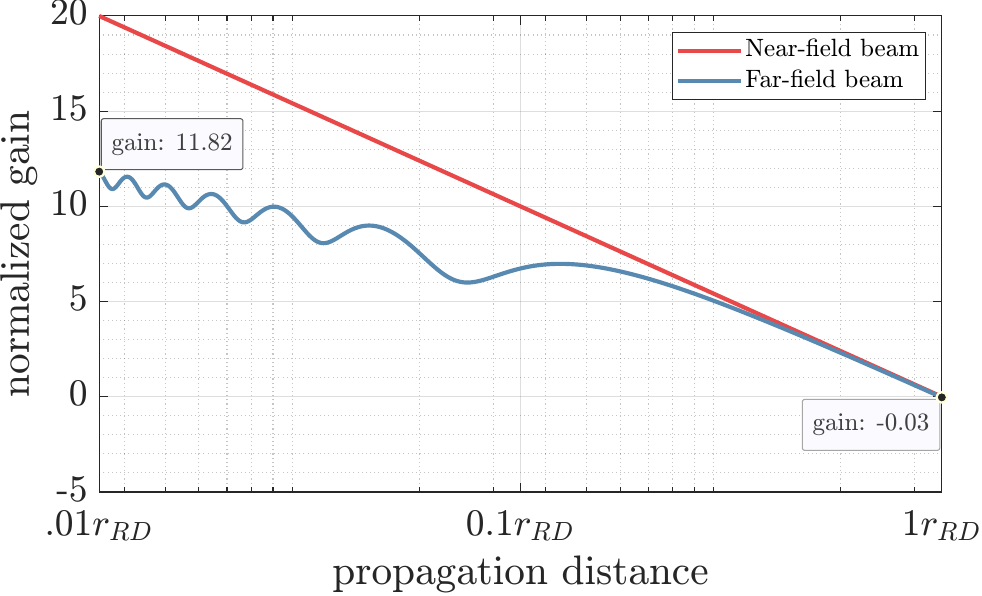}
\setlength{\belowcaptionskip}{-15pt}
\caption{Normalized gain of near and far-field beams in the near-field.}
\label{fig2}
\end{figure}

\subsection{ Is Far-field Beam Training Feasible in the Near-field?}
One major concern when deploying far-field beam training in the near-field is the reduction in far-field beamforming gain, denoted as $\mathcal{G}_\mathrm{ff} = \abs{\mathbf{a}^\mathsf{H}(\theta) \mathbf{h}(\theta, r)}$, where $\mathbf{h}(\theta, r)$ represents the near-field channel. Notably, the far-field beam suffers a $\unit[3]{dB}$ loss at the EBRD and experiences additional degradation at distances closer than the EBRD. The path loss factor, $g = \left(\lambda/4\pi r\right)$, which varies inversely with the distance $r$, mitigates the reduced beamforming gain at near-field ranges.

Fig. \ref{fig2} illustrates the normalized gain including path loss for near-field and far-field beams within the near-field communication range. The far-field beam yields approximately $\unit[12]{dB}$ higher gain at $0.01 \RD$ compared to the distance at $\RD$. Therefore, despite the far-field beam's diminished beamforming gain in the near-field, the worst-case overall gain attained in the near-field surpasses the best-case gain achieved in the far-field. As shown in Fig. \ref{fig2}, the far-field beam provides reduced SNR compared to near-field beams. Therefore, far-field beams may be employed for beam training rather than data transmission. After the feasibility of using far-field beams, we now focus on the analysis of the received DFT beam pattern in the near-field. 

\section{DFT Beam Pattern Analysis}
To ensure that a signal reaches a prospective user regardless of its location, we need to transmit DFT beams uniformly distributed over the coverage area. A far-field UE typically experiences high gain from a single DFT beam and reports its corresponding index to the BS. In contrast, a near-field UE may observe high gain across multiple DFT beams, depending on its distance from the BS. This difference arises because the far-field signals exhibit a planar wavefront, while the near-field signals present a spherical wavefront, which can be seen as a superposition of multiple planar wavefronts. In the sequel, based on the DFT beam sweeping, we quantify the array gain experienced by the near-field UE.

We consider a BS serving a LoS near-field UE at an angle $\au$ and distance $\rf$. The array gain from each DFT beam directed at angle $\theta_n$ and observed by the near-field UE is written as
\small
\begin{equation} 
\begin{aligned}
&\mathcal{G}(\au,\rf;\theta_n) ={\left| { \mathbf{b} ^{\mathsf{H}} (\au,\rf) \mathbf{a} (\theta_n)} \right|}^2,\\
&\approx
\tfrac{1}{\NBS^2}\left|\sum_{-\NBS/2}^{\NBS/2} e^{j\tfrac{2\pi}{\lambda}\{nd\sin(\au)- \frac{1}{2\rf}n^2d^2\cos^2(\au)\}-j\tfrac{2\pi nd\sin\theta_n}{\lambda} } \right|^2,\\
&\stackrel{(c_1)}{=}
\tfrac{1}{\NBS^2}\left|\sum_{-\NBS/2}^{\NBS} e^{-j\pi \{ n^2(\tfrac{d\cos^{2}\au}{2\rf}) - n(\sin\au - \sin\theta_n)\} } \right|^2,
\label{eqn9}
\end{aligned}
\end{equation}
\normalsize
where ($c_1$) is simplified assuming $d=\lambda/2$. We further simplify the array gain function in \eqref{eqn9} in terms of Fresnel functions $\mathcal{C(\cdot)}$ and $\mathcal{S(\cdot)}$, whose arguments depend on the UE location ($\au,\rf$) and DFT beam angle $\theta_n$.

\begin{theorem}
The gain function observed by a near-field UE when illuminated with orthogonal DFT beams can be approximated as
\begin{equation}
\begin{aligned}
\mathcal{G}(\au,\rf;\theta_n) 
&\approx
\left|\frac{\overline{\mathcal{C}}\left(\gamma_{1}, \gamma_{2}\right)+j \overline{\mathcal{S}}\left(\gamma_{1}, \gamma_{2}\right)}{2 \gamma_{2}}\right|^2, 
\label{eqn10}
\end{aligned}
\end{equation}

$\overline{\mathcal{C}}\left(\gamma_{1}, \gamma_{2}\right) \equiv \mathcal{C}\left(\gamma_{1}+\gamma_{2}\right)-\mathcal{C}\left(\gamma_{1}-\gamma_{2}\right)$ and $\overline{\mathcal{S}}\left(\gamma_{1}, \gamma_{2}\right) \equiv$
 $\mathcal{S}\left(\gamma_{1}+\gamma_{2}\right)-\mathcal{S}\left(\gamma_{1}-\gamma_{2}\right)$, where $\gamma_{1}=\sqrt{\frac{\rf}{d\cos^{2}\au}}(\sin\theta_n-\sin\au) $ and $\gamma_{2}=\frac{\NBS}{2} \sqrt{\frac{d\cos^{2}\au}{\rf}}$.
\label{theorem3}
\end{theorem}

\begin{proof}
We can express the gain function in \eqref{eqn9} as $\frac{1}{\NBS^2} \left| \sum_{-\NBS/2}^{\NBS/2} e^{-j \pi(A_1n - A_2)^2 } \right|^2$
where $A_{1}=\sqrt{\frac{d\cos^{2}\au}{2 \rf}}$ and $A_{2}=\frac{1}{A_{1}}\left(\frac{\sin\au - \sin\theta_n}{2}\right)$. When $\NBS \rightarrow \infty$, summation can be approximated as Riemann integral as follows
\begin{equation}
\begin{aligned}
&\mathcal{G}(\au,\rf;\theta_n) {\approx} \left|\tfrac{1}{\NBS} \int_{-\tfrac{\NBS}{2}}^{\tfrac{\NBS}{2}} e^{-j \pi\left(A_{1} n-A_{2}\right)^{2}} \mathrm{~d}n \right|^2, \\
 &\stackrel{\left(c_2\right)}{=} \left|\tfrac{1}{\NBS \sqrt{2} A_{1}} \int_{-\tfrac{1}{\sqrt{2}} A_{1} \NBS-\sqrt{2} A_{2}}^{\tfrac{1}{\sqrt{2}} A_{1} \NBS-\sqrt{2} A_{2}} e^{j \pi \frac{1}{2} t^{2}} \mathrm{~d}t \right|^2, 
\label{eqn11}
 \end{aligned}
\end{equation}
where $\left(c_2\right)$ is obtained by letting $A_{1} n-A_{2}=\frac{1}{\sqrt{2}} t$. Then, based on the properties of Fresnel integrals \cite{zhang2023mixed}, we obtain \eqref{eqn10}.
\end{proof}
\begin{figure}[t]
\centering
\includegraphics[width = 0.8\linewidth]
{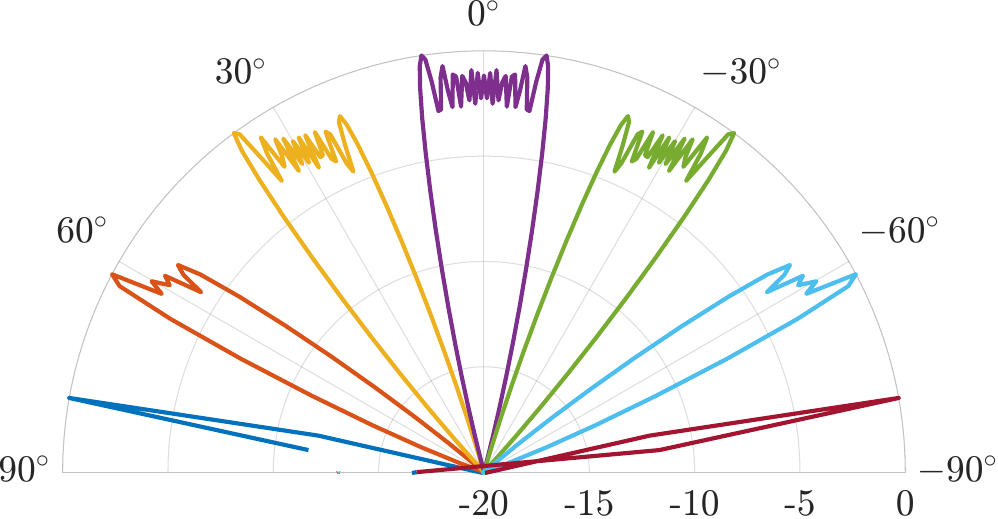}
\setlength{\belowcaptionskip}{-15pt}
\caption{The angular spread changes with the angle and reduces as the beam is directed toward larger angles away from the boresight.}
\label{fig3}
\end{figure}
For a given user, $\gamma_2$ is constant, and the gain function only varies with $\gamma_1$ that depends on $\theta_n$. Furthermore, $\mathcal{G}(\au,\rf;\theta_n)$ is symmetric with respect to $\gamma_1$ such that $\mathcal{G}(\gamma_1) = \mathcal{G}(-\gamma_1)$. Hence, high gain values in the beam pattern exist around the angle $\au$ of the UE. The angles where DFT beams provide a gain no less than $\unit[3]{dB}$ below the maximum, is referred to as the angular spread $\Omega_\mathrm{3dB}$ and defined as
\begin{equation}
\begin{aligned}
\Omega_\mathrm{3dB} \stackrel{\Delta}{=} \left\{\theta_n \mid \mathcal{G}\left( \au,\rf;\theta_n\right)>0.5 \max _{\theta_n} \mathcal{G}\left( \au,\rf;\theta_n\right)\right\}. 
\label{eqn12}
\end{aligned}
\end{equation}
To gain further insight, we plot the gain in \eqref{eqn10} at different UE angles and ranges as depicted in Figs. \ref{fig3} and \ref{fig4} respectively, and highlight the following observations:
\begin{enumerate}[label=(\roman*)]
    \item \textbf{ Angular spread $\Omega_\mathrm{3dB}$ decreases at endfire directions.} The angular spread across UE angles $[-80^\circ , -60^\circ , -30^\circ , 0^\circ, 30^\circ , 60^\circ , 80^\circ ]$ for a fixed range of $.01\RD$ is illustrated in Fig. \ref{fig3}, where it can be observed that $\Omega_\mathrm{3dB}$ decreases at large angles. Specifically, the angular spread is maximum at boresight and decreases at off-boresight angles. Also note that the pattern at $\theta_x$ is anti-symmetric to that at $\theta_{-x}$. The reduced angular spread at large angles is caused by the decreased effective array aperture, which results in a smaller phase difference between planar wave and spherical wave models. 
\item \textbf{Angular spread $\Omega_\mathrm{3dB}$ increases with a decrease in UE range.} 
Fig. \ref{fig4} depicts the angular spread across UE ranges and for a fixed UE angle ($\au = 0^\circ$). Observe that the angular spread varies with range, specifically widening at shorter distances. This is an important finding as it can be exploited to estimate the range of a near-field UE. 

\item \textbf{Angular spread $\Omega_\mathrm{3dB}$ exists only within the EBRD region ($\rf < \frac{\RD \cos^2\theta}{10}$) of the near-field}. From Fig. \ref{fig4}, it can be observed that $\Omega_\mathrm{3dB}$ diminishes when the user distance is greater than $\frac{\RD}{10}$. Likewise, in Fig. \ref{fig3}, $\Omega_\mathrm{3dB}$ approaches the far-field at large angles, i.e., at $\au = 80^\circ$. Therefore, the DFT beams can be utilized to estimate UE location only if the UE is positioned within the EBRD region.
\end{enumerate}

\begin{figure}[t]
\centering
\includegraphics[width = .62\linewidth]{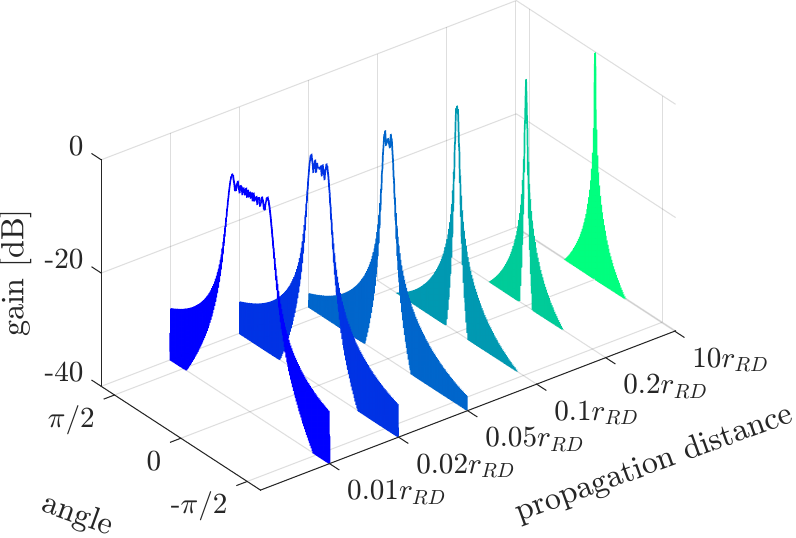}
\setlength{\belowcaptionskip}{-15pt}
\caption{The angular spread varies with range and decreases as the beam is focused at greater distances from the UM-MIMO array.}
\label{fig4}
\end{figure}

\section{Beam Training and Parameters Estimation}
In this section, we present the proposed beam training protocol and estimation procedure. The angle and range parameters of the UE are estimated by utilizing the DFT beam pattern analysis discussed in the previous section.
 
\subsection{Beam Training}
UM-MIMO transmits $\psi$ beams from the predefined DFT codebook to illuminate the coverage area of the BS, which is given by
\begin{equation}
 \mathbf{\Phi}^{\mathrm{DFT}} = [\boldsymbol{\phi}_1,\cdots,\boldsymbol{\phi}_n,\cdots, \boldsymbol{\phi}_\psi], 
 \label{eqn13}
\end{equation}
where $\boldsymbol{\phi}_n$ is given by \eqref{eqn2}.
Given the angular beamwidth $\theta_\mathrm{3dB} = \frac{2}{\NBS\cos\au}$ and coverage area defined by $\mathcal{A} =\abs{ \sin\theta_\mathrm{max} - \sin\theta_\mathrm{min}}$, the number of beams is $\psi = \frac{\mathcal{A}}{\theta_\mathrm{3dB}}$. Note that $\cos\au$ accounts for the beam broadening factor at large angles.
 
\subsection{Parameter Estimation based on Correlative Interferometry} 
Correlative interferometry is a well-established technique for determining the direction of an emitter. In traditional far-field settings, an antenna array undergoes a 360-degree rotation, during which the phases of the incoming signal from each direction are sampled and stored for each antenna element in a lookup table. Subsequently, the direction of the signal is estimated by correlating the received signal phases with the precomputed lookup table. Building on the principles of correlative interferometry, with specific modifications, we propose a CI-DFT-based range estimation algorithm tailored for near-field users located within the EBRD region.

As discussed in the previous section, the angular spread in the near-field provides information about the user’s location in terms of range and angle. To leverage this, we propose generating $\mathcal{G}(\au,\rf;\theta_n)$ for various ranges and angles based on \eqref{eqn10} and storing the corresponding $\Omega_{\mathrm{3dB}}$ values in a lookup table $\boldsymbol{\mathcal{K}}$. Each value of $\Omega_{\mathrm{3dB}}$ is unique for a specific combination of angle and distance, and is stored in $\boldsymbol{\mathcal{K}}$.

\subsubsection{Generate Lookup Table}
The procedure for generating the lookup table involves computing a polar grid and subsequently estimating the angular spread at each polar point. In this polar grid, angles are sampled uniformly, while range points are updated based on beamdepth. Beam depth limits maximum range resolution, thus determining the selection of range samples. The implementation of the lookup table generation is summarized in Algorithm \ref{algo_pcodebook}, whose main steps are briefly explained as follows:

\begin{figure}[ht]
    \centering
    \begin{subfigure}[t]{\linewidth}
        \centering
\includegraphics[width=.8\linewidth]{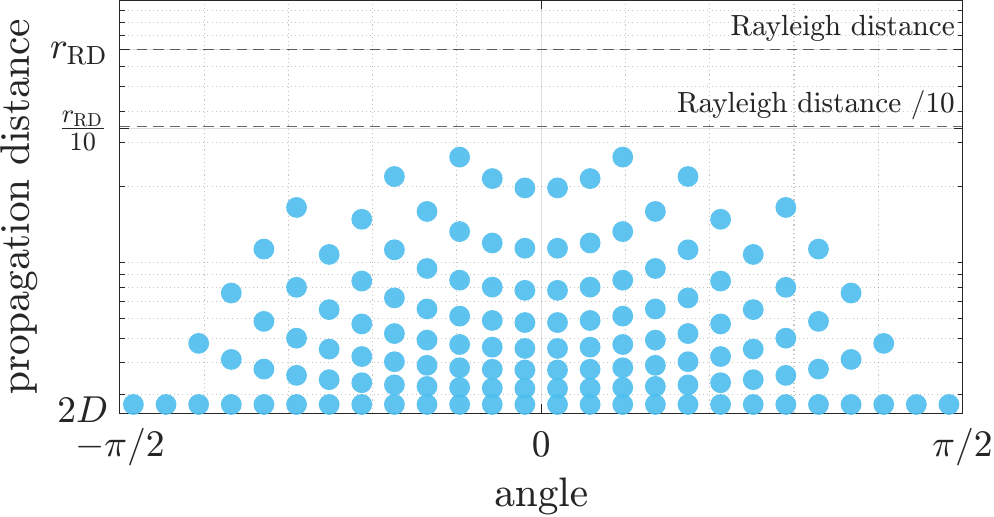} 
        \caption{}
                \vspace{.5em} 
        \label{fig_polargrid}
    \end{subfigure}
    \begin{subfigure}[t]{\linewidth}
        \centering
        \includegraphics[width=.8\linewidth]{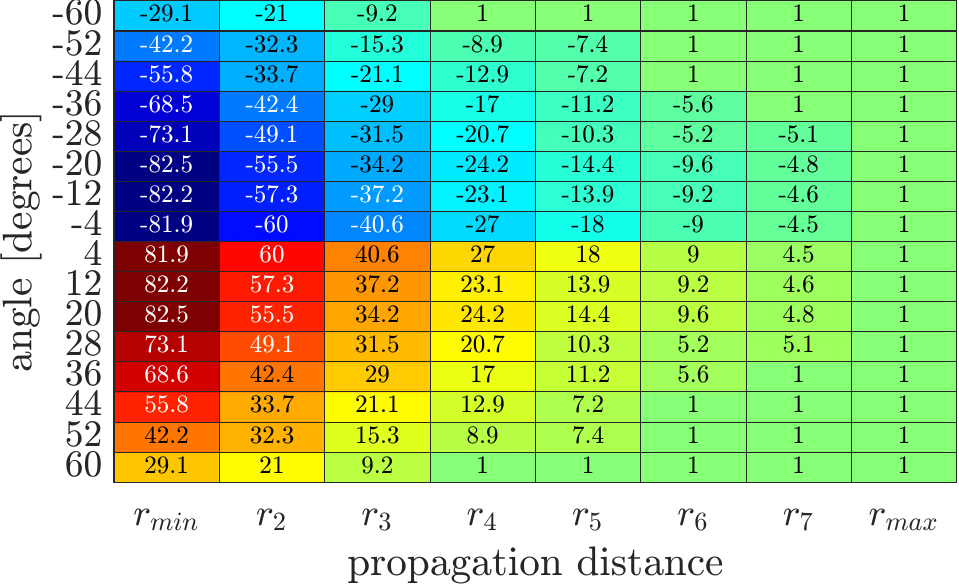} 
        \caption{}
        \label{fig5}
    \end{subfigure}
    \setlength{\belowcaptionskip}{-15pt}
    \caption{(a) Grid of polar points depicting uniform angle and non-uniform range sampling. (b) The lookup table shows that the angular spread increases at closer ranges, and decreases at large angles away from the boresight. Here, $r_\mathrm{min} = 2D$, $r_\mathrm{max} = \text{EBRD}$.}
\end{figure}

After uniform angle sampling in line 3, the minimum range is selected as $2D$, where $D$ represents the antenna aperture. This lower range limit is chosen because the array gain degrades significantly below this distance. For each angle $\theta_i$, as long as the beamdepth is shorter than the EBRD, we repeat steps from line 7-13 as follows: 

In line 8, for each combination of angle and distance $(\theta_i, r_{s,i})$, the DFT beam gain is computed using \eqref{eqn10}. Then, in line 9, the 3 dB angular spread $\Omega_\mathrm{3dB}$ is determined based on $\boldsymbol{\mathcal{G}}(\theta_i, r_{s,i}) = \mathcal{G}(\theta_i, r_{s,i}; \theta_n)$, $\forall$ $n \in [1, 2, \dots, \NBS]$. The calculated angular spread is then stored in $\boldsymbol{\mathcal{K}}(i,s)$. In line 11, the beamdepth $r_\mathrm{BD}$ is computed using \eqref{eqn8}, and in line 12, the next range sample is obtained by adding $r_\mathrm{BD}$ to the current range sample. We repeat steps from lines 7-13, $\forall$ $\theta \in [-\frac{\pi}{2} \ \frac{\pi}{2}]$, to generate the complete lookup table $\boldsymbol{\mathcal{K}}$. 

Fig. \ref{fig_polargrid} illustrates a polar grid characterized by uniform angular spacing and non-uniform radial sampling, for a frequency of $\unit[28]{GHz}$ and $\NBS = 256$. Figure \ref{fig5} depicts a sample lookup table for identical parameters. A heatmap of a limited number of values is displayed to convey the main concept. The variation in angular spread across range and angle follows the trend described earlier in Section IV. It is worth noting that the dynamic range across angles is small compared to the range dimension. For angle estimation, a coarse angle can be estimated by taking the median of the angular spread. The lookup table, customized to the specific parameters of the BS, is pre-generated and stored on the UE device.
 
\subsubsection{Angle and Range estimation}
After beam sweeping, based on the received DFT beam pattern, the UE measures angular spread $\hat{\Omega}_{\mathrm{3dB}}$ by \eqref{eqn12}. The UE then correlates $\hat{\Omega}_{\mathrm{3dB}}$ with the predefined lookup table $\boldsymbol{\mathcal{K}}(\theta_i,r_j)$ to estimate the angle $\hat{\theta}$ and range $\hat{r}$ as 
\begin{equation}
 \hat{\theta}, \hat{r} = \arg\max_{\theta_i, r_j}\left[\cos\left(\hat{\Omega}_\mathrm{3dB} - \mathcal{K}(\theta_i, r_j)\right)\right]. 
 \label{eqn14}
\end{equation}
After the estimation process, the UE feeds back the estimated angle and range $(\hat{\theta}, \hat{r})$ to the BS.

\IncMargin{0.5em}
\begin{algorithm}[t!]\footnotesize
\caption{Lookup Table for CI-DFT}\label{algo_pcodebook}
\SetKwInOut{Input}{Input}
\SetKwInOut{Output}{Output}
\DontPrintSemicolon
\SetKwFunction{CPG}{Compute Lookup Table }
\renewenvironment{algomathdisplay}
\opthanginginout 

\Input{$\lambda, \NBS, D, \RD$}

\vspace{2pt}
\SetKwProg{myproc}{Procedure}{}{end}
\myproc{\CPG}{
\For{$i = 1$ \KwTo $\NBS$}{

 $\theta_i \leftarrow \frac{-\NBS\pi+2i\pi}{2\NBS}$
 
 $\rf \leftarrow 2D$
 
 ${r}_{{s,i}}$ $\leftarrow$ $\{\emptyset\}$, $s \leftarrow1$
 
\While{$r_\mathrm{BD} \leq\frac{\RD}{10}\cos^2{\theta_i}$}{
 ${r}_{{s,i}} \leftarrow \rf$

 $\boldsymbol{\mathcal{G}}(\theta_i,r_{s,i}) \leftarrow \text{Ref } \ \eqref{eqn10}$ 

 $\Omega_\mathrm{3dB}(\theta_i,r_{s,i}) \leftarrow \boldsymbol{\mathcal{G}}(\theta_i,r_{s,i})$ \tcp*{Ref \eqref{eqn12}} 

$\boldsymbol{\mathcal{K}}(i,s) \leftarrow \Omega_\mathrm{3dB}(\theta_i,r_{s,i}) $

 ${r_\mathrm{BD} \leftarrow \text{Ref } \ \eqref{eqn8}}$
 
 ${\rf \leftarrow \rf + r_\mathrm{BD}}$
 
 $s \leftarrow s+1$
 }
 
 }
\KwRet{$\boldsymbol{\mathcal{K}}$}
 }

\end{algorithm}\DecMargin{0.5em}

\section{Simulation Results}
In this section, numerical results are presented to validate the efficacy of the proposed CI-DFT-based algorithm. The BS employing UM-MIMO array is equipped with $\NBS = 256$ antenna elements, and the operating frequency is $\unit[28]{GHz}$. For this setup, Rayleigh distance is $\unit[350]{m}$ and EBRD equals $35\cos^2\theta$. All the presented results are based on $1000$ iterations. We compare with the following benchmark schemes:

\begin{itemize} 
\item \textbf{Perfect CSI}: CSI is assumed to be available at the BS, providing an upper bound for the average rate performance.

\item \textbf{Near-field hierarchical beam training \cite{lu2023hierarchical}}: In this scheme, the near-field hierarchical codebook consists of $K$ levels of sub-codebooks, dividing the near-field region into $N_x$ and $N_y$ grids in the x and y directions, respectively.
\item \textbf{Far-field beam training}: This is the classical beam training designed for far-field scenarios.
\item \textbf{Exhaustive search}: This scheme performs an exhaustive search over the polar-domain codebook, involving a two-dimensional search across both the angular and range domains.
\end{itemize} 

We compare the different schemes based on the rate $R = \log_2\left(1+ \frac{P_t\NBS\abs{\mathbf{b}^\mathsf{H}(\theta,r)\mathbf{w}^{\star}}^2}{\sigma^2}\right)$, where $P_t$ is the transmitted power, $\sigma^2$ is the noise variance, and $\mathbf{w}^{\star}$ is the chosen best codeword.

Fig. \ref{fig6} illustrates the average rate performance as a function of propagation distance. The UE range varies from $\unit[3]{m}$ to $\frac{\RD}{10} =\unit[35]{m}$ while the spatial angle $\sin\theta$ is randomly distributed as $\mathcal{U} [-\frac{\sqrt{3}}{2}, \frac{\sqrt{3}}{2}]$. Under the assumption of perfect calibration, the proposed CI-DFT closely matches the exhaustive search scheme while significantly reducing training overhead. The proposed CI-DFT outperforms other benchmark schemes by effectively utilizing the DFT beam pattern to estimate the user range and angle. To simulate array calibration errors, random phase values were introduced within the interval $\mathcal{U}[0, \pi/8]$, along with random amplitude values drawn from $\mathcal{U} [0, 1]$. Phase errors have a more detrimental impact on performance compared to amplitude errors. 

The average rate performance against SNR is shown in Fig. \ref{fig7}. The user is uniformly distributed between $ \mathcal{U}[\unit[3]{m}, \unit[35]{m}]$. It is observed that the average rate performance of all schemes improves with increasing SNR. Notably, the performance of the proposed CI-DFT closely approaches that of the exhaustive scheme.
 
Table \ref{table:1} summarizes the comparison of training overhead across all schemes. The angular beamwidth $\theta_\mathrm{3dB}$, calculated as $\frac{2}{\NBS}$, is $.0078$. To adequately cover the spatial region \(\mathcal{U} = [\sin\theta_{\min}, \sin\theta_{\max}]\), where \(\sin\theta_{\min} = -\frac{\sqrt{3}}{2}\) and \(\sin\theta_{\max} = \frac{\sqrt{3}}{2}\), a total of \(\psi = 222\) angular beams are utilized, along with \(S = 8\) distance samples.
 Consequently, the training overhead for the near-field exhaustive scheme amounts to $1776$. In contrast, the near-field hierarchical approach with parameters $N_x = 25$, $N_y = 25$, and $K=2$, results in a reduced overhead of $1250$. The training overhead of the proposed CI-DFT is equivalent to that of the far-field scheme, amounting to $222$. The proposed CI-DFT reduces the training overhead by a factor of $8$ (i.e. 87.5\%) compared to the exhaustive search scheme while still achieving equivalent performance.

\begin{figure}[t]
\centering
\includegraphics[scale = 0.45]{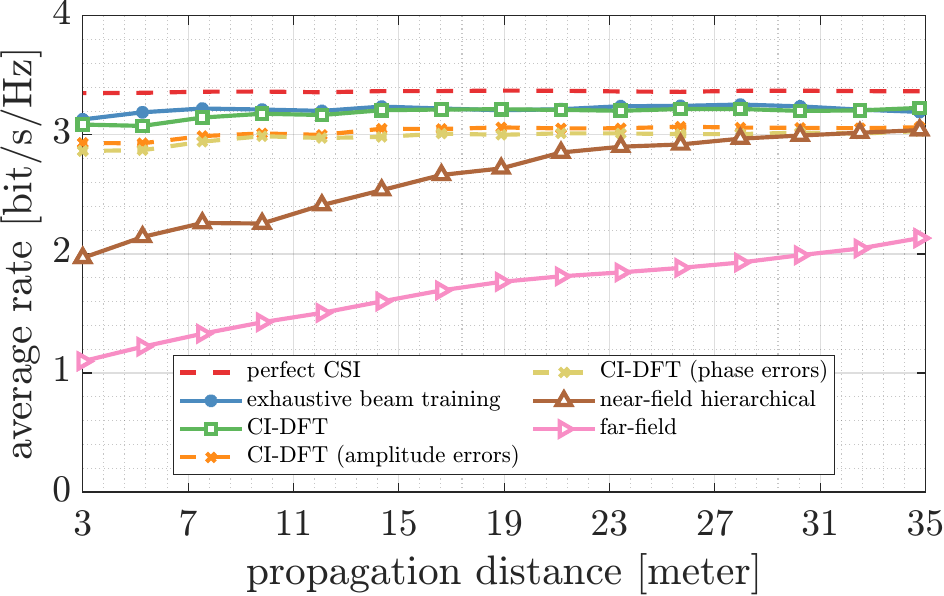}
\caption{Average rate performance vs. distance.}
\setlength{\belowcaptionskip}{-20pt}
\label{fig6}
\end{figure}

\begin{figure}[t]
\centering
\includegraphics[scale = 0.45]{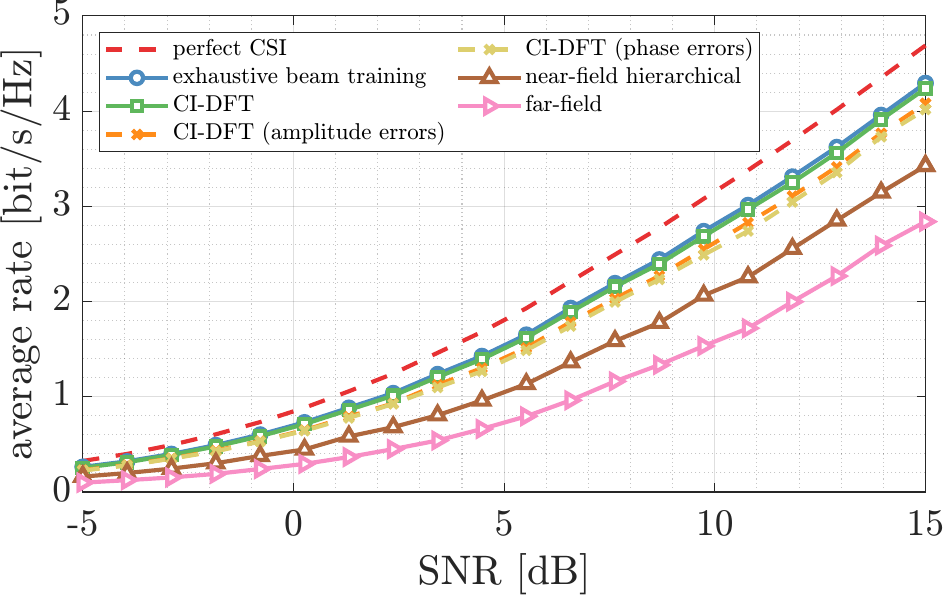}
\caption{Average rate performance vs. SNR.}
\label{fig7}
\end{figure}

\begin{table}[t]\footnotesize
\setlength{\tabcolsep}{10pt}
\renewcommand{\arraystretch}{1.25}
\centering
\caption{Comparison of Training overhead}
\begin{tabular}{|l|l|l|}
\hline
\textbf{Algorithm} & \textbf{Training Overhead} & \textbf{Value}\\
\hline
Near-Field Exhaustive & $\psi S$ & $1776$\\
Near-Field hierarchical & $\sum_{k}^{K} {N_{x}^{(k)}} {N_{y}^{(k)}} S$ & $1250$\\
Far-Field & $\psi $ & $222$\\
Proposed CI-DFT & $\psi$ & $222$\\
 \hline
\end{tabular}
\vspace{-10pt}
\label{table:1} 
\end{table}

\section{Conclusion}
In this work, we have proposed a near-field beam training algorithm based on a far-field codebook. We have first identified the coverage areas of both near-field and far-field codewords and demonstrated that far-field codewords provide the necessary SNR for near-field beam training. Additionally, we have shown that the DFT-based beam pattern received in the near-field contains both range and angle information. Leveraging this, we have proposed a correlation interferometry-based algorithm to estimate the angle and the range parameters. A potential future direction is to extend this approach to other planar geometries. 

\bibliographystyle{IEEEtran}
\bibliography{IEEEabrv,my2bib}

\begin{thebibliography}{10}
\providecommand{\url}[1]{#1}
\csname url@samestyle\endcsname
\providecommand{\newblock}{\relax}
\providecommand{\bibinfo}[2]{#2}
\providecommand{\BIBentrySTDinterwordspacing}{\spaceskip=0pt\relax}
\providecommand{\BIBentryALTinterwordstretchfactor}{4}
\providecommand{\BIBentryALTinterwordspacing}{\spaceskip=\fontdimen2\font plus
\BIBentryALTinterwordstretchfactor\fontdimen3\font minus \fontdimen4\font\relax}
\providecommand{\BIBforeignlanguage}[2]{{%
\expandafter\ifx\csname l@#1\endcsname\relax
\typeout{** WARNING: IEEEtran.bst: No hyphenation pattern has been}%
\typeout{** loaded for the language `#1'. Using the pattern for}%
\typeout{** the default language instead.}%
\else
\language=\csname l@#1\endcsname
\fi
#2}}
\providecommand{\BIBdecl}{\relax}
\BIBdecl

\bibitem{rappaport2019wireless}
T.~S. Rappaport, Y.~Xing, O.~Kanhere, S.~Ju, A.~Madanayake, S.~Mandal, A.~Alkhateeb, and G.~C. Trichopoulos, ``Wireless communications and applications above 100 {GHz}: Opportunities and challenges for 6{G} and beyond,'' \emph{IEEE access}, vol.~7, pp. 78\,729--78\,757, 2019.

\bibitem{han2021hybrid}
C.~Han, L.~Yan, and J.~Yuan, ``Hybrid beamforming for terahertz wireless communications: Challenges, architectures, and open problems,'' \emph{IEEE Wireless Communications}, vol.~28, no.~4, pp. 198--204, 2021.

\bibitem{giordani2018tutorial}
M.~Giordani, M.~Polese, A.~Roy, D.~Castor, and M.~Zorzi, ``A tutorial on beam management for {3GPP} {NR} at {mmWave} frequencies,'' \emph{IEEE Communications Surveys \& Tutorials}, vol.~21, no.~1, pp. 173--196, 2018.

\bibitem{OS_Asmaa}
A.~Abdallah, A.~Celik, M.~M. Mansour, and A.~M. Eltawil, ``Multi-agent deep reinforcement learning for beam codebook design in {RIS}-aided systems,'' \emph{IEEE Trans. Wireless Commun.}, vol.~23, no.~7, pp. 7983--7999, 2024.

\bibitem{asmaa2024near}
\BIBentryALTinterwordspacing
A.~Abdallah, A.~Hussain, A.~Celik, and A.~M. Eltawil, ``Exploring frontiers of polar-domain codebooks for near-field channel estimation and beamfocusing,'' \emph{Submitted to IEEE SPM}, 2024. [Online]. Available: \url{https://repository.kaust.edu.sa/handle/10754/700927}
\BIBentrySTDinterwordspacing

\bibitem{cui2022channel}
M.~Cui and L.~Dai, ``Channel estimation for extremely large-scale {MIMO}: Far-field or near-field?'' \emph{{IEEE} Trans. Commun.}, vol.~70, no.~4, pp. 2663--2677, Jan. 2022.

\bibitem{zhang2022fast}
Y.~Zhang, X.~Wu, and C.~You, ``Fast near-field beam training for extremely large-scale array,'' \emph{IEEE Wireless Commun. Lett.}, vol.~11, no.~12, pp. 2625--2629, 2022.

\bibitem{lu2023hierarchical}
Y.~Lu, Z.~Zhang, and L.~Dai, ``Hierarchical beam training for extremely large-scale {MIMO}: From far-field to near-field,'' \emph{IEEE Transactions on Communications}, vol.~72, no.~4, pp. 2247--2259, Apr. 2024.

\bibitem{wu2023two}
C.~Wu, C.~You, Y.~Liu, L.~Chen, and S.~Shi, ``Two-stage hierarchical beam training for near-field communications,'' \emph{IEEE Transactions on Vehicular Technology}, 2023.

\bibitem{cui2024near}
M.~Cui and L.~Dai, ``Near-field wideband beamforming for extremely large antenna arrays,'' \emph{IEEE Transactions on Wireless Communications}, 2024.

\bibitem{bjornson2021primer}
E.~Bj{\"o}rnson, {\"O}.~T. Demir, and L.~Sanguinetti, ``A primer on near-field beamforming for arrays and reconfigurable intelligent surfaces,'' in \emph{Proc. Asilomar Conf. Signals, Systems and Computers (Asilomar)}.\hskip 1em plus 0.5em minus 0.4em\relax Pacific Grove, CA, USA: IEEE, Mar. 2021, pp. 105--112.

\bibitem{hussain2024near}
A.~Hussain, A.~Abdallah, and A.~M. Eltawil, ``Near-field channel estimation for ultra-massive {MIMO} antenna array with hybrid architecture,'' in \emph{2024 IEEE Wireless Communications and Networking Conference (WCNC)}.\hskip 1em plus 0.5em minus 0.4em\relax IEEE, 2024, pp. 1--6.

\bibitem{ahmed2024nearfield}
------, ``Redefining polar boundaries for near-field channel estimation for ultra-massive {MIMO} antenna array,'' \emph{Submitted to IEEE TWC}, 2024.

\bibitem{zhang2023mixed}
Y.~Zhang, C.~You, L.~Chen, and B.~Zheng, ``Mixed near-and far-field communications for extremely large-scale array: An interference perspective,'' \emph{IEEE Communications Letters}, 2023.

\end{thebibliography}
\end{document}